\documentclass[12pt,a4paper]{article}
\usepackage[top=2.5cm,bottom=2.7cm,left=2.2cm,right=2.2cm]{geometry}
\usepackage{graphicx,latexsym,amsmath,amsfonts,amssymb,amsthm,amscd}
\usepackage{hyperref}

\newcounter{statement}

\theoremstyle{plain}
\newtheorem{lemma}[statement]{Lemma}
\newtheorem{theorem}[statement]{Theorem}

\def\1{{\boldsymbol 1}}       %
\def\0{{\boldsymbol 0}}       %
\def\bP{{\boldsymbol P}}      %
\def\bQ{{\boldsymbol Q}}      %
\def\bR{{\boldsymbol R}}      %
\def\bU{{\boldsymbol U}}      %
\def\bV{{\boldsymbol V}}      %
\def\bW{{\boldsymbol W}}      %
\def\bX{{\boldsymbol X}}      %
\def\cG{{\mathcal G}}         %
\def\cO{{\mathcal O}}         %
\def\cR{{\mathcal R}}         %
\def\red{{\rm red}}           %
\def\tr{{\rm tr}}             %
\def\diag{{\rm diag}}         %
\def\rd{{\rm d}}              %
\def\ri{{\rm i}}              %
\def\Re{{\rm Re}}             %
\def\C{\mathbb{C}}            %
\def\N{\mathbb{N}}            %
\def\R{\mathbb{R}}            %
\def\T{{\mathbb T}}           %
\def\UN{{\rm U}}              %
\def\fL{\mathfrak{L}}         %

\begin{document}
\begin{center}
{\large\bf
On the derivation of Darboux form for the action-angle
dual of trigonometric BC${}_\textit{n}$ Sutherland system}
\end{center}

\vspace{0.2cm}

\begin{center}
T.F.~G{\"o}rbe\\

\bigskip

Department of Theoretical Physics, University of Szeged\\
Tisza Lajos krt 84-86, H-6720 Szeged, Hungary\\
e-mail: tfgorbe@physx.u-szeged.hu

\bigskip
\end{center}

\vspace{0.2cm}

\begin{abstract}
Recently Feh\'er and the author have constructed the action-angle
dual of the trigonometric BC${}_n$ Sutherland system via Hamiltonian
reduction. In this paper\footnote{Contribution to the proceedings
of the 22nd International Conference on ``Integrable Systems and Quantum
Symmetries'' (ISQS-22, Prague, June 2014).} a reduction-based
calculation is carried out to verify canonical Poisson bracket
relations on the phase space of this dual model. Hence the material
serves complementary purposes whilst it can also be regarded as a
suitable modification of the hyperbolic case previously sorted out by Pusztai.
\end{abstract}

\medskip
{\bf Keywords:} {\em Integrable systems}; {\em Hamiltonian reduction};
{\em Darboux form}

\medskip
{\bf MSC2010:} 14H70; 37J15; 53D20

\medskip
{\bf PACS number:} 02.30.Ik

\newpage

\section{Introduction}
\label{sec:1}

The integrable one-dimensional many-body systems
of Calogero, Moser, and Sutherland and generalized versions of them have proven to be a
fruitful source of both diverse physical applications and connections between seemingly
distant areas of mathematics. For details, see e.g.
\cite{RUIJSENAARS1,ETINGOF,POLYCHRONAKOS}.
Among the numerous aspects of these models their duality relations are rather
interesting. Two Liouville integrable many-body Hamiltonian systems $(M,\omega,H)$ and
$(\tilde M,\tilde\omega,\tilde H)$ with Darboux coordinates $q,p$ and $\lambda,\vartheta$,
respectively, are said to be duals of each other if there is a global symplectomorphism
$\cR\colon M\to\tilde M$ of the phase spaces, which exchanges the canonical coordinates
with the action-angle variables for the Hamiltonians. Practically, this means that
$H\circ\cR^{-1}$ depends only on $\lambda$, while $\tilde H\circ\cR$ only on $q$.
In more detail, $q$ are the particle positions for $H$ and action variables for
$\tilde H$, and similarly, $\lambda$ are the positions of particles modelled by the
Hamiltonian $\tilde H$ and action variables for $H$.

A notable work has been done by Ruijsenaars \cite{RUIJSENAARS2,RUIJSENAARS3}
in constructing action-angle duality maps for models with rational, hyperbolic, and
trigonometric potentials associated with the A${}_n$ root system. Many of these
dualities have been interpreted in terms of Hamiltonian reduction
\cite{FOCK-GORSKY-NEKRASOV-RUBTSOV,FEHER-KLIMCIK}.

The suspected existence of action-angle duality between models related other root systems
has been confirmed by Pusztai \cite{PUSZTAI2} proving the hyperbolic BC${}_n$ Sutherland
\cite{OLSHANETSKY-PERELOMOV} and the rational BC${}_n$ Ruijsenaars\,--\,Schneider\,--\,van
Diejen (RSvD) \cite{VANDIEJEN} systems to be in
duality.

In a recent paper by Feh\'er and the author \cite{FEHER-GORBE} earlier results
\cite{PUSZTAI2,FEHER-AYADI} have been generalized to obtain a new dual pair involving
the trigonometric BC${}_n$ Sutherland system. This was achieved by applying Hamiltonian
reduction to the cotangent bundle $T^\ast\UN(2n)$ with respect to the symmetry group
$G_+\times G_+$ with $G_+\simeq\UN(n)\times\UN(n)$. The systems in duality arose as two
cross sections of the orbits of the symmetry group in the level surface of the momentum
map since these cross sections were identified with the phase spaces of the trigonometric
BC${}_n$ Sutherland and a rational BC${}_n$ RSvD-type systems. The aim of this paper is to
provide detailed calculations proving that under this identification the coordinates
$\lambda,\vartheta$ -- introduced on a dense submanifold of the phase space of the dual
model -- are canonical (Darboux) coordinates as stated in \cite{FEHER-GORBE}.

Section \ref{sec:2} is a selective review of \cite{FEHER-GORBE} devoted to establishing
context and introducing necessary notations for succeeding calculations. The core of
the paper is Section \ref{sec:3} which contains a series of lemmas culminating in the
main result. Concluding the paper, Section \ref{sec:4} gives a brief discussion of the
outcome and its relation to other cases considered formerly.

\section{Context and notations}
\label{sec:2}

Choose an arbitrary positive integer, $n$. Let $G$ and $\cG$ denote the unitary group
U($2n$) and its Lie algebra, respectively. The Lie algebra $\cG$ can be equipped with
the Ad-invariant bilinear form
\begin{equation}
\langle\cdot,\cdot\rangle\colon\cG\times\cG\to\R,\quad
(Y_1,Y_2)\mapsto\langle Y_1,Y_2\rangle=\tr(Y_1Y_2),
\label{1}
\end{equation}
which allows one to identify $\cG$ with its dual space $\cG^\ast$ in the usual manner.
The cotangent bundle $T^\ast G$ can be trivialized using left-translations
\begin{equation}
T^\ast G\cong G\times\cG^\ast\cong G\times\cG
=\{(y,Y)\mid y\in G,\ Y\in\cG\}.
\label{2}
\end{equation}
Then the canonical symplectic form of $T^\ast G$ can be written as
$\Omega^{T^*G}=-\rd\langle y^{-1}\rd y,Y\rangle$, and it can be evaluated locally
according to the formula
\begin{equation}
\Omega^{T^\ast G}_{(y,Y)}(\Delta y\oplus\Delta Y,\Delta'y\oplus\Delta'Y)
=\langle y^{-1}\Delta y,\Delta'Y\rangle
-\langle y^{-1}\Delta'y,\Delta Y\rangle
+\langle[y^{-1}\Delta y,y^{-1}\Delta'y],Y\rangle,
\label{3}
\end{equation}
where $\Delta y\oplus\Delta Y,\Delta'y\oplus\Delta'Y\in T_{(y,Y)} T^*G$
are arbitrary tangent vectors at a point $(y,Y)\in T^*G$. By introducing the $2n\times 2n$
Hermitian, unitary matrix
\begin{equation}
C=\begin{bmatrix}
\0_n&\1_n\\
\1_n&\0_n
\end{bmatrix}\in G,
\label{4}
\end{equation}
where $\1_n$ and $\0_n$ denote the identity and null matrices of size $n$, respectively,
an involutive automorphism of $G$ can be defined as conjugation with $C$
\begin{equation}
\Gamma\colon G\to G,\quad
y\mapsto\Gamma(y)=CyC^{-1}.
\label{5}
\end{equation}
The fix-point subgroup of $\Gamma$ in $G$ is
\begin{equation}
G_+=\{y\in G\mid\Gamma(y)=y\}\cong\UN(n)\times\UN(n).
\label{6}
\end{equation}
Let $\Gamma$ stand for the induced involution of the Lie algebra $\cG$, too.
Hence $\cG$ can be decomposed as
\begin{equation}
\cG=\cG_+\oplus\cG_-,\quad Y=Y_++Y_-,
\label{7}
\end{equation}
where $\cG_\pm$ are the eigenspaces of $\Gamma$
corresponding to the eigenvalues $\pm 1$, respectively.

In \cite{FEHER-GORBE} a reduction of $T^\ast G$ based on the symmetry group
$G_+ \times G_+$ was performed by using the shifting trick of symplectic reduction
\cite{ORTEGA-RATIU}. For that a coadjoint orbit of the symmetry group must be prepared.
To any vector $V\in\C^{2n}$ that satisfies $CV+V=0$ associate an element
$\upsilon_{\mu,\nu}^\ell(V)$ of $\cG_+$ by the definition
\begin{equation}
\upsilon_{\mu,\nu}^\ell(V)=\ri\mu\big(VV^\dag-\1_{2n}\big)+\ri(\mu-\nu)C,
\label{8}
\end{equation}
where $\mu,\nu\in\R$ are real parameters.
The set
\begin{equation}
\cO^\ell=
\big\{\upsilon^\ell\in\cG_+\mid
\exists\ V\in\C^{2n},\ V^\dag V=2n,\ CV+V=0,\
\upsilon^\ell=\upsilon_{\mu,\nu}^\ell(V)
\big\}
\label{9}
\end{equation}
represents a coadjoint orbit of $G_+$ of dimension $2(n-1)$.
Let $\cO^r:=\{\upsilon^r\}$ denote the one-point coadjoint orbit of $G_+$
containing the element $\upsilon^r=-\ri\kappa C$ with some constant $\kappa\in\R$ and
consider
\begin{equation}
\cO=\cO^\ell\oplus\cO^r\subset\cG_+\oplus\cG_+\cong(\cG_+\oplus\cG_+)^\ast,
\label{10}
\end{equation}
which is a coadjoint orbit of $G_+ \times G_+$.
The initial phase space for symplectic reduction is
\begin{equation}
P=T^\ast G\times\cO
\quad\text{with the symplectic form}\quad
\Omega=\Omega^{T^\ast G}+\Omega^{\cO},
\label{11}
\end{equation}
where $\Omega^{\cO}$ is the Kirillov\,--\,Kostant\,--\,Souriau symplectic form on the
coadjoint orbit $\cO$.

\noindent For any point $x=(y,Y,\upsilon^\ell,\upsilon^r)\in P$ and smooth functions
$f,f'\in C^\infty(P)$
\begin{equation}
\Omega_x((\bX_f)_x,(\bX_{f'})_x)
=\Omega^{T^*G}_{(y,Y)}(\Delta y\oplus\Delta Y,\Delta'y\oplus\Delta'Y)
+\langle[D_{\upsilon^\ell},D'_{\upsilon^\ell}],\upsilon^\ell\rangle,
\label{12}
\end{equation}
where $(\bX_f)_x=\Delta y\oplus\Delta Y\oplus\Delta\upsilon^\ell\oplus 0$,
$(\bX_{f'})_x=\Delta'y\oplus\Delta'Y\oplus\Delta'\upsilon^\ell\oplus 0\in T_x P$
and $\Delta\upsilon^\ell=[D_{\upsilon^\ell},\upsilon^\ell]$,
$\Delta'\upsilon^\ell=[D'_{\upsilon^\ell},\upsilon^\ell]$ with some $\cG_+$-valued
$D_{\upsilon^\ell}^{\phantom{'}},D'_{\upsilon^\ell}$. The natural symplectic action of
$G_+\times G_+$ on $P$ is defined by
\begin{equation}
\Phi_{(g_L,g_R)}(y,Y,\upsilon^\ell,\upsilon^r)=
\big(g_L^{\phantom{1}}yg_R^{-1},
     g_R^{\phantom{1}}Yg_R^{-1},
     g_L^{\phantom{1}}\upsilon^\ell g_L^{-1},
     \upsilon^r \big).
\label{13}
\end{equation}
The corresponding  momentum map $J\colon P\to\cG_+\oplus\cG_+$ is given by the formula
\begin{equation}
J(y,Y,\upsilon^\ell,\upsilon^r)=
\big((yYy^{-1})_++\upsilon^\ell\big)\oplus\big(-Y_++\upsilon^r\big).
\label{14}
\end{equation}
The reduced phase space is
\begin{equation}
P_\red=J^{-1}(0)/(G_+\times G_+),
\label{15}
\end{equation}
which is a smooth symplectic manifold.

One of the main results in \cite{FEHER-GORBE} was the construction of a semi-global
cross section of symmetry group orbits in the momentum constraint surface $J^{-1}(0)$,
that is a model of the reduced phase space \eqref{15}. This was done by solving the
momentum equation $J(y,Y,\upsilon^\ell,\upsilon^r)=\0_{2n}\oplus\0_{2n}$ through the
diagonalization of the ($\cG_-$)-part of the Lie algebra component. In particular,
the following matrix similarity was demonstrated
\begin{equation}
Y\sim\ri h(\lambda)\Lambda(\lambda)h(\lambda)^{-1},
\label{16}
\end{equation}
where $\Lambda(\lambda)=\diag(\lambda,-\lambda)$ with
$\lambda=(\lambda_1,\ldots,\lambda_n)\in\R^n$ subject to
$\lambda_1>\cdots>\lambda_n>|\kappa|$ and $h(\lambda)$ is the unitary matrix
\begin{equation}
h(\lambda)=\begin{bmatrix}
\alpha(\diag(\lambda))&\beta(\diag(\lambda))\\
-\beta(\diag(\lambda))&\alpha(\diag(\lambda))
\end{bmatrix},
\label{17}
\end{equation}
with the real functions $\alpha(x),\beta(x)$ defined on the interval
$[|\kappa|,\infty)\subset\R$ by the formulae
\begin{equation}
\alpha(x)=\frac{\sqrt{x+\sqrt{x^2-\kappa^2}}}{\sqrt{2x}},\quad
\beta(x)=\kappa\frac{1}{\sqrt{2x}}\frac{1}{\sqrt{x+\sqrt{x^2-\kappa^2}}},
\label{18}
\end{equation}
if $\kappa \neq 0$. For $\kappa=0$, set $h(\lambda)=\1_{2n}$.
This approach enables one to define the smooth map
\begin{equation}
\fL\colon P_0\to \R^n,\quad
(y,Y,\upsilon^\ell,\upsilon^r)\mapsto\lambda,
\label{19}
\end{equation}
which descends to a smooth map $\fL_\red\colon P_\red\to\R^n$.
The image of the constraint surface $J^{-1}(0)$ under the map $\fL$ \eqref{19}
turned out to be the closure of the domain
\begin{equation}
C_2=\bigg\{\lambda\in\R^n\bigg|
\begin{matrix}\lambda_a-\lambda_{a+1}>2\mu,\\
(a=1,\ldots,n-1)\end{matrix}
\quad\text{and}\quad
\lambda_n>\nu \bigg\}.
\label{20}
\end{equation}
Introduce the vector $F\in\C^{2n}$ by the formulae
\begin{equation}
\begin{split}
&F_{a}=\bigg[1-\frac{\nu}{\lambda_a}\bigg]^{\tfrac{1}{2}}
\prod_{\substack{b=1\\(b\neq a)}}^n
\bigg[1-\frac{2\mu}{\lambda_a-\lambda_b}\bigg]^{\tfrac{1}{2}}
\bigg[1-\frac{2\mu}{\lambda_a+\lambda_b}\bigg]^{\tfrac{1}{2}}, \quad a\in\{1,\ldots,n\},\\
&F_{n+a}=e^{\ri\vartheta_a}\bigg[1+\frac{\nu}{\lambda_a}\bigg]^{\tfrac{1}{2}}
\prod_{\substack{b=1\\(b\neq a)}}^n
\bigg[1+\frac{2\mu}{\lambda_a-\lambda_b}\bigg]^{\tfrac{1}{2}}
\bigg[1+\frac{2\mu}{\lambda_a+\lambda_b}\bigg]^{\tfrac{1}{2}}.
\end{split}
\label{21}
\end{equation}
and the $2n\times 2n$ matrices $A(\lambda, \vartheta)$ and $B(\lambda, \vartheta)$ by
\begin{equation}
A_{j,k}(\lambda,\vartheta)=\frac{2\mu F_j\overline{(CF)}_k-
2(\mu-\nu)C_{j,k}}{2\mu-\Lambda_j+\Lambda_k},\quad
j,k\in\{1,\ldots,2n\},
\label{22}
\end{equation}
and
\begin{equation}
B(\lambda,\vartheta)=-\big(h(\lambda)A(\lambda,\vartheta)h(\lambda)\big)^\dagger.
\label{23}
\end{equation}
These are unitary matrices satisfying $\Gamma(A)=A^{-1}$,
$\Gamma(B)=B^{-1}$. The matrix $B$ can be diagonalized using some $\eta\in G_+$
\begin{equation}
B=\eta\,\diag(\exp(2\ri q),\exp(-2\ri q))\eta^{-1},
\label{24}
\end{equation}
where $q=q(\lambda,\vartheta)\in\R^n$ is unique and subject to $\pi/2>q_1>\cdots>q_n>0$.
Relying on \eqref{24} set
\begin{equation}
y(\lambda,\vartheta)=\eta\,\diag(\exp(\ri q),\exp(-\ri q))\eta^{-1},
\label{25}
\end{equation}
and introduce the vector $V(\lambda,\vartheta)\in\C^{2n}$ by
\begin{equation}
V(\lambda,\vartheta)=y(\lambda,\vartheta)h(\lambda)F(\lambda,\vartheta).
\label{26}
\end{equation}
It was also shown in \cite{FEHER-GORBE} that $V+CV=0$ and $|V|^2=2n$ ensuring that 
$\upsilon^\ell_{\mu,\nu}(V)\in\cO^\ell$ \eqref{9}.

Theorem 4.1 of \cite{FEHER-GORBE} claims that the set
\begin{equation}
\tilde S^0:=\{(y(\lambda,\vartheta),\ri h(\lambda)\Lambda(\lambda)h(\lambda)^{-1},
\upsilon^\ell_{\mu,\nu}(V(\lambda,\vartheta)),
\upsilon^r)\mid(\lambda,e^{\ri \vartheta})\in
C_2\times\T^n\}.
\label{27}
\end{equation}
is contained in the constraint surface $J^{-1}(0)$ and provides a cross-section for the
$G_+\times G_+$-action restricted to $\fL^{-1}(C_2)\subset J^{-1}(0)$. In particular,
$C_2\subset\fL(J^{-1}(0))$ and $\tilde S^0$ intersects every gauge orbit in
$\fL^{-1}(C_2)$ precisely in one point. Since the elements of $\tilde S^0$ are
parametrized by $C_2\times\T^n$ in a smooth and bijective manner, the following
identifications were gained
\begin{equation}
\fL^{-1}_\red(C_2)\simeq\tilde S^0\simeq C_2\times\T^n.
\label{28}
\end{equation}
Let $\tilde \sigma_0$ denote the tautological injection
\begin{equation}
\tilde\sigma_0\colon\tilde S^0\to P.
\label{29}
\end{equation}
This way $C_2 \times \T^n$ yields a model of an open submanifold $\fL^{-1}(C_2)$ of
$P_\red$ corresponding to the open submanifold $\fL^{-1}(C_2) \subset J^{-1}(0)$ was
obtained. The purpose of this paper is to show that the pull-back
$\tilde\sigma_0^\ast(\Omega)$ of the symplectic form $\Omega$ \eqref{11} is
\begin{equation}
\tilde\sigma_0^\ast(\Omega)=\sum_{a=1}^n\rd\lambda_a\wedge\rd\vartheta_a
\label{30}
\end{equation}
by computing the Poisson brackets
\begin{equation}
\{\lambda_a,\lambda_b\},\quad
\{\lambda_a,\vartheta_b\},\quad
\{\vartheta_a,\vartheta_b\},\qquad
a,b\in\{1,\ldots,n\}.
\label{31}
\end{equation}
Now, consider the reduced functions $f_j^\red=\tilde\sigma_0^\ast(f_j)$ for some
$f_j\in C^\infty(P)^{G_+\times G_+}$ ($j=1,2$). Then the definition of symplectic
reduction implies
\begin{equation}
\tilde\sigma_0^\ast(\{f_1,f_2\})=\{f_1^\red,f_2^\red\},
\label{32}
\end{equation}
where the Poisson bracket on the left-hand-side is computed on $(P,\Omega)$ \eqref{11}.
The idea is to extract the required Poisson brackets in \eqref{31} from equality \eqref{32}
applied to various choices of $f_1,f_2$.
Note that $\{f_1,f_2\}=\Omega(\bX_{f_2},\bX_{f_1})$
with the corresponding Hamiltonian vector fields.

\section{Calculation of Poisson brackets}
\label{sec:3}

The following verification is an appropriate adaptation of an argument presented by
Pusztai in \cite{PUSZTAI1} which since has been applied in the simpler case of A${}_n$
root system in \cite{AYADI-FEHER-GORBE}. Differences between these earlier results
and the calculations below are highlighted in the Discussion.

Consider the following families of real-valued smooth functions on the phase space $P$
\eqref{11}
\begin{eqnarray}
\varphi_m(y,Y,\upsilon^\ell,\upsilon^r):=\dfrac{1}{m}\Re\big(\tr(Y^m)\big),
\quad m\in\N,\label{33}\\
\chi_k(y,Y,\upsilon^\ell,\upsilon^r):=\Re\big(\tr(Y^ky^{-1}Z(\upsilon^\ell)yC)\big),
\quad k\in\N_0,\label{34}
\end{eqnarray}
where $Z(\upsilon^\ell)=(\ri\mu)^{-1}\upsilon_{\mu,\nu}^\ell(V)+\1_N-(1-\nu/\mu)C=VV^\dag$.
The corresponding reduced functions on $\tilde S^0$ are
\begin{equation}
\varphi_m^\red(\lambda,\vartheta)=
\begin{cases}
0,&\text{if}\ m\ \text{is odd},\\
\displaystyle(-1)^{\tfrac{m}{2}}\frac{2}{m}\sum_{j=1}^n\lambda_j^m,
&\text{if}\ m\ \text{is even},
\end{cases}
\label{35}
\end{equation}
and
\begin{equation}
\chi_k^\red(\lambda,\vartheta)=
\begin{cases}
\displaystyle(-1)^{\tfrac{k+1}{2}} 2\sum_{a=1}^n\lambda_a^k
\bigg[1-\frac{\kappa^2}{\lambda_a^2}\bigg]^{\tfrac{1}{2}}
|X_a|\sin(\vartheta_a),&\mbox{if }k\mbox{ is odd},\\
\displaystyle(-1)^{\tfrac{k}{2}} 2\sum_{a=1}^n\lambda_a^k
\bigg[1-\frac{\kappa^2}{\lambda_a^2}\bigg]^{\tfrac{1}{2}}
|X_a|\cos(\vartheta_a)
-\kappa\lambda_a^{k-1}\big(|F_a|^2-|F_{n+a}|^2\big),&\mbox{if }k\mbox{ is even},
\end{cases}
\label{36}
\end{equation}
where
\begin{equation}
X_a=F_a\overline{F}_{n+a}
=e^{-\ri\vartheta_a}\bigg[1-\frac{\nu^2}{\lambda_a^2}\bigg]^{\tfrac{1}{2}}
\prod_{\substack{b=1\\(b\neq a)}}^n
\bigg[1-\frac{4\mu^2}{(\lambda_a-\lambda_b)^2}\bigg]^{\tfrac{1}{2}}
\bigg[1-\frac{4\mu^2}{(\lambda_a+\lambda_b)^2}\bigg]^{\tfrac{1}{2}}.
\label{37}
\end{equation}
Now let us take an arbitrary point $x=(y,Y,\upsilon^\ell,\upsilon^r)\in P$ and an arbitrary
tangent vector $\delta x=\delta y\oplus\delta Y\oplus\delta\upsilon^\ell\oplus 0\in T_x P$.
The derivative of $\varphi_m$ can be easily obtained and has the form
\begin{equation}
(\rd\varphi_m)_x(\delta x)=
\begin{cases}
0,&\mbox{if }m\mbox{ is odd},\\
\langle Y^{m-1},\delta Y\rangle,&\mbox{if }m\mbox{ is even}.
\end{cases}
\label{38}
\end{equation}
The derivative of $\chi_k$ can be written as
\begin{equation}
\begin{split}
(\rd\chi_k)_x(\delta x)=&
\bigg\langle\dfrac{\big[[Y^k,C]_\pm,y^{-1}Z(\upsilon^\ell)y\big]}{2},y^{-1}
\delta y\bigg\rangle\\
&+\bigg\langle\sum_{j=0}^{k-1}\dfrac{Y^{k-j-1}[y^{-1}Z(\upsilon^\ell)y,C]_\pm
Y^j}{2},\delta Y\bigg\rangle\\
&+\bigg\langle\dfrac{y[C,Y^k]_\pm y^{-1}+Cy[C,Y^k]_\pm
y^{-1}C}{4\ri\mu},\delta\upsilon^\ell\bigg\rangle,
\end{split}
\label{39}
\end{equation}
where $[A,B]_\pm:=AB\pm BA$ with the sign of $(-1)^k$.
The Hamiltonian vector field of $\varphi_m$ is
\begin{equation}
(\bX_{\varphi_m})_x
=\Delta y\oplus\Delta Y\oplus\Delta\upsilon^\ell\oplus 0
=yY^{m-1}\oplus 0\oplus 0\oplus 0,
\label{40}
\end{equation}
while the Hamiltonian vector field corresponding to $\chi_k$ is
\begin{equation}
(\bX_{\chi_k})_x
=\Delta' y\oplus\Delta' Y\oplus\Delta' \upsilon^\ell\oplus 0,
\label{41}
\end{equation}
where
\begin{alignat}{3}
\Delta'y&=\dfrac{y}{2}\sum_{j=0}^{k-1}Y^{k-j-1}[y^{-1}Z(\upsilon^\ell)y,C]_\pm Y^j,
\label{42}\\
\Delta'Y&=\dfrac{1}{2}\big[[Y^k,y^{-1}Z(\upsilon^\ell)y]_\pm,C\big],
\label{43}\\
\Delta'\upsilon^\ell&=\dfrac{1}{4\ri\mu}\big[\big(y[C,Y^k]_\pm y^{-1}+Cy[C,Y^k]_\pm
y^{-1}C\big),\upsilon^\ell\big].
\label{44}
\end{alignat}

\begin{lemma}
$\{\lambda_a,\lambda_b\}=0$ for any $a,b\in\{1,\ldots,n\}$.
\label{lemma:1}
\end{lemma}

\begin{proof}
Using \eqref{38} one has $\{\varphi_m,\varphi_l\}\equiv 0$ for any $m,l\in\N$
which implies that $\{\varphi_m^\red,\varphi_l^\red\}\equiv 0$. Let $m,l\in\N$ be
arbitrary even numbers. Direct calculation of the Poisson bracket
$\{\varphi_m^\red,\varphi_l^\red\}$ using \eqref{35} and the Leibniz rule results in the
formula
\begin{equation}
\{\varphi_m^\red,\varphi_l^\red\}=(-1)^{\tfrac{m+l}{2}}4\sum_{a,b=1}^n
\lambda_a^{m-1}\{\lambda_a,\lambda_b\}\lambda_b^{l-1}.
\label{45}
\end{equation}
By introducing the $n\times n$ matrices
\begin{equation}
\bP_{a,b}:=\{\lambda_a,\lambda_b\}\quad
\text{and}\quad
\bU_{a,b}:=\lambda_a^{2b-1},\qquad
a,b\in\{1,\ldots,n\}
\label{46}
\end{equation}
and choosing $m$ and $l$ from the set $\{1,\ldots,2n\}$, the equation
$\{\varphi_m^\red,\varphi_l^\red\}\equiv 0$ can be cast into the matrix equation
\begin{equation}
(-1)^{\tfrac{m+l}{2}}\bU^\dag\bP\bU=\0_n.
\label{47}
\end{equation}
Since $\bU$ is an invertible Vandermonde-type matrix it follows from \eqref{47}
that $\bP=\0_n$ which reads as $\{\lambda_a,\lambda_b\}=0$ for all $a,b\in\{1,\ldots,n\}$.
\end{proof}

\begin{lemma}
$\{\lambda_a,\vartheta_b\}=\delta_{a,b}$ for any $a,b\in\{1,\ldots,n\}$.
\label{lemma:2}
\end{lemma}

\begin{proof}
By choosing two even numbers, $k$ and $m$, and calculating the Poisson bracket
$\{\chi_k,\varphi_m\}$ at an arbitrary point $x=(y,Y,\upsilon^\ell,\upsilon^r)\in P$
the results \eqref{40}-\eqref{44} imply that
\begin{equation}
\{\chi_k,\varphi_m\}(x)=\chi_{k+m-1}(x)
+\frac{1}{2}\tr\big((Y^kCY^{m-1}-Y^{m-1}CY^k)y^{-1}Z(\upsilon^\ell)y\big).
\label{48}
\end{equation}
The computation of the reduced form of \eqref{48} shows that
\begin{equation}
\{\chi_k^\red,\varphi_m^\red\}=2\chi_{k+m-1}^\red.
\label{49}
\end{equation}
By utilizing \eqref{35}, \eqref{36} and the result of the previous lemma one can write
the l.h.s. of \eqref{49} as
\begin{equation}
\{\chi_k^\red,\varphi_m^\red\}=(-1)^{\tfrac{k+m}{2}}
4\sum_{b=1}^n\lambda_b^k\bigg[1-\frac{\kappa^2}{\lambda_b^2}\bigg]^{\tfrac{1}{2}}
|X_b(\lambda)|\sin(\vartheta_b)\sum_{a=1}^n\{\lambda_a,\vartheta_b\}\lambda_a^{m-1}.
\label{50}
\end{equation}
Now, returning to equation \eqref{49} together with \eqref{50}
one can obtain the following equivalent form
\begin{equation}
\sum_{b=1}^n
\lambda_b^k\bigg[1-\frac{\kappa^2}{\lambda_b^2}\bigg]^{\tfrac{1}{2}}
|X_b(\lambda)|\sin(\vartheta_b)
\bigg(\sum_{a=1}^n\{\lambda_a,\vartheta_b\}\lambda_a^{m-1}-\lambda_b^{m-1}\bigg)
=0.
\label{51}
\end{equation}
By introducing the $n\times n$ matrices
\begin{equation}
\bV_{b,d}:=\bigg[1-\frac{\kappa^2}{\lambda_b^2}\bigg]^{\tfrac{1}{2}}
|X_b(\lambda)|\sin(\vartheta_b)\bigg(\sum_{a=1}^n\{\lambda_a,\vartheta_b\}
\lambda_a^{2d-1}-\lambda_b^{2d-1}\bigg),
\quad b,d\in\{1,\ldots,n\}
\label{52}
\end{equation}
and using the Vandermonde-type matrix $\bU$ defined in \eqref{46}
one is able to write \eqref{51} into the matrix equation $\bU^\dag\bV=\0_n$.
Since $\bU$ is invertible $\bV=\0_n$ and therefore in the dense subset of
$C_2\times\T^n$ where $\sin(\vartheta_b)\neq 0$ the following holds
\begin{equation}
\sum_{a=1}^n\{\lambda_a,\vartheta_b\}\lambda_a^{m-1}-\lambda_b^{m-1}=0,
\quad\forall\,b\in\{1,\ldots,n\}.
\label{53}
\end{equation}
With the matrices $\bU$ and
\begin{equation}
\bQ_{b,a}:=\{\lambda_a,\vartheta_b\},\quad
a,b\in\{1,\ldots,n\}
\label{54}
\end{equation}
equation \eqref{53} can be written equivalently as $\bQ\bU-\bU=\0_n$,
which immediately implies that $\bQ=\1_n$. Due to the continuity of Poisson bracket
$\bQ=\1_n$ must hold for every point in $C_2\times\T^n$, therefore one has
$\{\lambda_a,\vartheta_b\}=\delta_{a,b}$ for all $a,b\in\{1,\ldots,n\}$.
\end{proof}

\begin{lemma}
$\{\vartheta_a,\vartheta_b\}=0$ for any $a,b\in\{1,\ldots,n\}$.
\label{lemma:3}
\end{lemma}

\begin{proof}
Let $k$ and $l$ be two arbitrarily chosen odd integers, and set $f=\chi_l$ and $f'=\chi_k$
in \eqref{12}. First, one can calculate the Poisson bracket $\{\chi_k^\red,\chi_l^\red\}$
indirectly, that is, work out the Poisson bracket
$\{\chi_k,\chi_l\}=\Omega(\bX_{\chi_l},\bX_{\chi_k})$
explicitly and restrict it to the gauge \eqref{27}. The first term on the right-hand side
of equation \eqref{12}, namely $\langle y^{-1}\Delta y,\Delta'Y\rangle$ can be written as
\begin{equation}
\begin{split}
\langle y^{-1}\Delta y,\Delta'Y\rangle=&
(-1)^{\tfrac{k+l+2}{2}} 2\,l\sum_{a=1}^n
\lambda_a^{k+l-1}\bigg[1-\frac{\kappa}{\lambda_a^2}\bigg]
|X_a(\lambda)|^2\sin(2\vartheta_a)\\
&(-1)^{\tfrac{k+l+2}{2}} 2\sum_{\substack{a,b=1\\(a\neq b)}}^n
\lambda_a^k\lambda_b^l
\bigg[1-\frac{\kappa^2}{\lambda_a^2}\bigg]^{\tfrac{1}{2}}
\bigg[1-\frac{\kappa^2}{\lambda_b^2}\bigg]^{\tfrac{1}{2}}
|X_a||X_b|
\frac{\sin(\vartheta_a-\vartheta_b)}
{\lambda_a+\lambda_b}\\
&(-1)^{\tfrac{k-l+2}{2}} 2\sum_{\substack{a,b=1\\(a\neq b)}}^n
\lambda_a^k\lambda_b^l
\bigg[1-\frac{\kappa^2}{\lambda_a^2}\bigg]^{\tfrac{1}{2}}
\bigg[1-\frac{\kappa^2}{\lambda_b^2}\bigg]^{\tfrac{1}{2}}
|X_a||X_b|
\frac{\sin(\vartheta_a+\vartheta_b)}
{\lambda_a-\lambda_b}.
\end{split}
\label{55}
\end{equation}
Due to antisymmetry in the indices the second term can be gained
by interchanging $k$ and $l$
\begin{equation}
\begin{split}
\langle y^{-1}\Delta' y,\Delta Y\rangle=&
(-1)^{\tfrac{k+l+2}{2}} 2\,k\sum_{a=1}^n
\lambda_a^{k+l-1}\bigg[1-\frac{\kappa}{\lambda_a^2}\bigg]
|X_a(\lambda)|^2\sin(2\vartheta_a)\\
&(-1)^{\tfrac{k-l+2}{2}} 2\sum_{\substack{a,b=1\\(a\neq b)}}^n
\lambda_a^k\lambda_b^l
\bigg[1-\frac{\kappa^2}{\lambda_a^2}\bigg]^{\tfrac{1}{2}}
\bigg[1-\frac{\kappa^2}{\lambda_b^2}\bigg]^{\tfrac{1}{2}}
|X_a||X_b|
\frac{\sin(\vartheta_a-\vartheta_b)}
{\lambda_a+\lambda_b}\\
&(-1)^{\tfrac{k+l+2}{2}} 2\sum_{\substack{a,b=1\\(a\neq b)}}^n
\lambda_a^k\lambda_b^l
\bigg[1-\frac{\kappa^2}{\lambda_a^2}\bigg]^{\tfrac{1}{2}}
\bigg[1-\frac{\kappa^2}{\lambda_b^2}\bigg]^{\tfrac{1}{2}}
|X_a||X_b|
\frac{\sin(\vartheta_a+\vartheta_b)}
{\lambda_a-\lambda_b}.
\end{split}
\label{56}
\end{equation}
One can easily check that the third term in \eqref{12} vanishes.
The last term of \eqref{12} takes the form
\begin{equation}
\begin{split}
\langle[D_\upsilon,D'_\upsilon],\upsilon\rangle=
&(-1)^{\tfrac{k+l+2}{2}} 4\sum_{\substack{a,b=1\\(a\neq b)}}^n
\lambda_a^k\lambda_b^l
\bigg[1-\frac{\kappa^2}{\lambda_a^2}\bigg]^{\tfrac{1}{2}}
\bigg[1-\frac{\kappa^2}{\lambda_b^2}\bigg]^{\tfrac{1}{2}}
|X_a||X_b|
\frac{\sin(\vartheta_a-\vartheta_b)}
{\big(4\mu^2-(\lambda_a+\lambda_b)^2\big)(\lambda_a+\lambda_b)}\\
&(-1)^{\tfrac{k-l+2}{2}} 4\sum_{\substack{a,b=1\\(a\neq b)}}^n
\lambda_a^k\lambda_b^l
\bigg[1-\frac{\kappa^2}{\lambda_a^2}\bigg]^{\tfrac{1}{2}}
\bigg[1-\frac{\kappa^2}{\lambda_b^2}\bigg]^{\tfrac{1}{2}}
|X_a||X_b|
\frac{\sin(\vartheta_a+\vartheta_b)}
{\big(4\mu^2-(\lambda_a-\lambda_b)^2\big)(\lambda_a-\lambda_b)}.
\end{split}
\label{57}
\end{equation}
As a result of this indirect calculation one obtains the following expression for
$\{\chi_k^\red,\chi_l^\red\}$
\begin{equation}
\begin{split}
\{\chi_k^\red,\chi_l^\red\}&=
(-1)^{\tfrac{k-l+2}{2}} 2(k-l)\sum_{a=1}^n\lambda_a^{k+l-1}
\bigg[1-\frac{\kappa^2}{\lambda_a^2}\bigg]|X_a|^2\sin(2\vartheta_a)\\
&(-1)^{\tfrac{k+l+2}{2}} 16\mu^2\sum_{\substack{a,b=1\\(a\neq b)}}^n
\lambda_a^k\lambda_b^l
\bigg[1-\frac{\kappa^2}{\lambda_a^2}\bigg]^{\tfrac{1}{2}}
\bigg[1-\frac{\kappa^2}{\lambda_b^2}\bigg]^{\tfrac{1}{2}}
|X_a||X_b|
\frac{\sin(\vartheta_a-\vartheta_b)}
{\big(4\mu^2-(\lambda_a+\lambda_b)^2\big)(\lambda_a+\lambda_b)}\\
&(-1)^{\tfrac{k-l+2}{2}} 16\mu^2\sum_{\substack{a,b=1\\(a\neq b)}}^n
\lambda_a^k\lambda_b^l
\bigg[1-\frac{\kappa^2}{\lambda_a^2}\bigg]^{\tfrac{1}{2}}
\bigg[1-\frac{\kappa^2}{\lambda_b^2}\bigg]^{\tfrac{1}{2}}
|X_a||X_b|
\frac{\sin(\vartheta_a+\vartheta_b)}
{\big(4\mu^2-(\lambda_a-\lambda_b)^2\big)(\lambda_a-\lambda_b)}.
\end{split}
\label{58}
\end{equation}
One can also carry out a direct computation of $\{\chi_k^\red,\chi_l^\red\}$
by using basic properties of the Poisson bracket and the previous two lemmas
\begin{equation}
\begin{split}
\{\chi_k^\red,\chi_l^\red\}&=
(-1)^{\tfrac{k-l+2}{2}} 2(k-l)\sum_{a=1}^n\lambda_a^{k+l-1}
\bigg[1-\frac{\kappa^2}{\lambda_a^2}\bigg]|X_a|^2\sin(2\vartheta_a)\\
&(-1)^{\tfrac{k+l+2}{2}} 16\mu^2\sum_{\substack{a,b=1\\(a\neq b)}}^n
\lambda_a^k\lambda_b^l
\bigg[1-\frac{\kappa^2}{\lambda_a^2}\bigg]^{\tfrac{1}{2}}
\bigg[1-\frac{\kappa^2}{\lambda_b^2}\bigg]^{\tfrac{1}{2}}
|X_a||X_b|
\frac{\sin(\vartheta_a-\vartheta_b)}
{\big(4\mu^2-(\lambda_a+\lambda_b)^2\big)(\lambda_a+\lambda_b)}\\
&(-1)^{\tfrac{k-l+2}{2}} 16\mu^2\sum_{\substack{a,b=1\\(a\neq b)}}^n
\lambda_a^k\lambda_b^l
\bigg[1-\frac{\kappa^2}{\lambda_a^2}\bigg]^{\tfrac{1}{2}}
\bigg[1-\frac{\kappa^2}{\lambda_b^2}\bigg]^{\tfrac{1}{2}}
|X_a||X_b|
\frac{\sin(\vartheta_a+\vartheta_b)}
{\big(4\mu^2-(\lambda_a-\lambda_b)^2\big)(\lambda_a-\lambda_b)}\\
&(-1)^{\tfrac{k-l}{2}} 4\sum_{a,b=1}^n\lambda_a^k\lambda_b^l
\bigg[1-\frac{\kappa^2}{\lambda_a^2}\bigg]^{\tfrac{1}{2}}
\bigg[1-\frac{\kappa^2}{\lambda_b^2}\bigg]^{\tfrac{1}{2}}
|X_a||X_b|\cos(\vartheta_a)\cos(\vartheta_b)\{\vartheta_a,\vartheta_b\}.
\end{split}
\label{59}
\end{equation}
Now it is obvious that \eqref{58} and \eqref{59}
must be equal therefore the extra term must vanish
\begin{equation}
\sum_{a,b=1}^n\lambda_a^k\lambda_b^l
\bigg[1-\frac{\kappa^2}{\lambda_a^2}\bigg]^{\tfrac{1}{2}}
\bigg[1-\frac{\kappa^2}{\lambda_b^2}\bigg]^{\tfrac{1}{2}}
|X_a||X_b|\cos(\vartheta_a)\cos(\vartheta_b)\{\vartheta_a,\vartheta_b\}=0.
\label{60}
\end{equation}
By utilizing the $n\times n$ matrices
\begin{equation}
\bW_{a,b}=\lambda_a^b\bigg[1-\frac{\kappa^2}{\lambda_a^2}\bigg]^{\tfrac{1}{2}}
|X_a(\lambda)|\cos(\vartheta_a),\quad
\bR_{a,b}=\{\vartheta_a,\vartheta_b\},\quad
a,b\in\{1,\ldots,n\}
\label{61}
\end{equation}
one can reformulate \eqref{60} as the matrix equation
\begin{equation}
\bW^\dag\bR\,\bW=\0_n.
\label{62}
\end{equation}
Since $\bW$ is easily seen to be invertible in a dense subset of the phase space
$C_2\times\T^n$, eq. \eqref{62} and the continuity of Poisson bracket imply
$\bR=\0_n$ for the full phase space, i.e., $\{\vartheta_a,\vartheta_b\}=0$ for all
$a,b\in\{1,\ldots,n\}$.
\end{proof}

Lemmas \ref{lemma:1}, \ref{lemma:2}, and \ref{lemma:3} together imply the following
result of \cite{FEHER-GORBE}, whose proof was omitted in that paper to save space.

\begin{theorem}
The reduced symplectic structure on $\tilde S^0$ \eqref{27}, given by the
pull-back of $\Omega$ \eqref{11} by the map $\tilde\sigma_0$ \eqref{29},
has the canonical form
$\tilde\sigma_0^\ast(\Omega)=\sum_{a=1}^n\rd\lambda_a\wedge\rd\vartheta_a$.
\label{theorem:4}
\end{theorem}

\section{Discussion}
\label{sec:4}

In this paper an explicit derivation of the Darboux form \eqref{30} was given.
The Poisson bracket relations
\begin{equation}
\{\lambda_a,\lambda_b\}=0,\quad
\{\lambda_a,\vartheta_b\}=\delta_{a,b},\quad
\{\vartheta_a,\vartheta_b\}=0,\qquad
a,b\in\{1,\ldots,n\}
\label{63}
\end{equation}
were proved in Lemmas \ref{lemma:1}, \ref{lemma:2}, and \ref{lemma:3}, respectively.
As a consequence Theorem \ref{theorem:4} was proved.

As mentioned before the method used in this paper has been previously applied to the
analogous hyperbolic models associated with the C${}_n$ \cite{PUSZTAI1} and A${}_n$
\cite{FEHER-KLIMCIK,AYADI-FEHER-GORBE} root systems. In \cite{PUSZTAI2} the hyperbolic
BC${}_n$ case has been settled by ``an almost verbatim computation as in the C${}_n$
case''. In fact, a careful comparison of corresponding equations shows subtle differences
as a result of the dissimilar characteristics of the underlying systems. For example,
most of the expressions in Section \ref{sec:3} contain factors with the parameter $\kappa$
which reflects the BC${}_n$ feature. As one would expect taking the limit $\kappa\to0$
turns these formulae into the ones seen in the C${}_n$ case. The trigonometric nature of
the considered systems can be accounted for another difference when minor complications
occur in Lemmas \ref{lemma:2} and \ref{lemma:3} due to the appearance of trigonometric
functions. These issues have been resolved by using density and continuity arguments.

\bigskip
\noindent{\bf Acknowledgements.}

The author is grateful to L\'aszl\'o Feh\'er for his valuable suggestions.
This work was supported in part by the EU and the State of Hungary,
co-financed by the European Social Fund in the framework of
T\'AMOP-4.2.4.A/2-11/1-2012-0001 `National Excellence Program' and by the
Hungarian Scientific Research Fund (OTKA) under the grant K-111697.

\end{document}